\documentclass{article}
\usepackage{ijcai18}

\usepackage{times}
\usepackage{xcolor}
\usepackage{soul}
\usepackage[utf8]{inputenc}
\usepackage[small]{caption}

\usepackage{amsmath}
\usepackage{amssymb}
\usepackage{latexsym}
\usepackage{graphicx}
\usepackage{bm}
\usepackage{nicefrac}
\usepackage{booktabs}
\usepackage{array}
\usepackage{multirow}
\usepackage{threeparttable}
\usepackage{makecell}
\usepackage[procnumbered,ruled,vlined,linesnumbered]{algorithm2e}
\usepackage{graphicx}

\input amssym.def
\newsymbol\rtimes 226F
\newfont{\nset}{msbm10}

\def\len#1{\left\lVert #1 \right\rVert}
\def\kh#1{\left( #1 \right)}

\def\ceil#1{\left\lceil #1 \right\rceil}

\def\defeq{\stackrel{\mathrm{def}}{=}}

\newcommand{\removelatexerror}{\let\@latex@error\@gobble}

\newcommand{\eps}{\epsilon}
\newcommand{\rea}{\mathbb{R}}
\newcommand{\LaplSolver}{\textsc{LaplSolve}}
\newcommand{\ExactSM}{\textsc{ExactSM}}
\newcommand{\ApproxiSM}{\textsc{ApproxiSM}}
\newcommand{\VReffComp}{\textsc{VReffComp}}

\newcommand\yy{\boldsymbol{\mathit{y}}}
\newcommand\zz{\boldsymbol{\mathit{z}}}
\newcommand\xx{\boldsymbol{\mathit{x}}}
\newcommand\bb{\boldsymbol{\mathit{b}}}
\newcommand\ee{\boldsymbol{\mathit{e}}}

\newcommand\LL{\bm{\mathit{L}}}
\renewcommand\AA{\boldsymbol{\mathit{A}}}
\newcommand\BB{\boldsymbol{\mathit{B}}}
\newcommand\JJ{\boldsymbol{\mathit{J}}}
\newcommand\DD{\boldsymbol{\mathit{D}}}
\newcommand\EE{\boldsymbol{\mathit{E}}}
\newcommand\PP{\boldsymbol{\mathit{P}}}
\newcommand\MM{\boldsymbol{\mathit{M}}}
\newcommand\QQ{\boldsymbol{\mathit{Q}}}
\newcommand\NN{\boldsymbol{\mathit{N}}}

\newcommand\Otil{\widetilde{O}}

\def\defeq{\stackrel{\mathrm{def}}{=}}
\def\trace#1{\mathrm{Tr} \left(#1 \right)}
\def\sizeof#1{\left|#1  \right|}

\def\expec#1#2{{\mathbb{E}}_{#1}\left[ #2 \right]}

\newtheorem{theo}{Theorem}[section]
\newtheorem{theorem}[theo]{Theorem}

\newtheorem{lemma}[theo]{Lemma}
\newtheorem{definition}[theo]{Definition}

\DontPrintSemicolon
\SetKw{KwAnd}{and}
\SetProcNameSty{textsc}
\SetFuncSty{textsc}
\SetKwInOut{Input}{Input\ \ \ \ }
\SetKwInOut{Output}{Output}

\title{Improving information centrality of a node in complex networks by adding edges\thanks{This work was supported by NSF.} }

\author{
Liren Shan$^1$,
Yuhao Yi$^1$,
Zhongzhi Zhang$^1$,
\\
$^1$ Shanghai Key Laboratory of Intelligent Information Processing, School of Computer Science, \\Fudan University, Shanghai 200433, China\\
13307130150@fudan.edu.cn,
15110240008@fudan.edu.cn,
zhangzz@fudan.edu.cn
}

\begin{document}

\maketitle

\begin{abstract}
The problem of increasing the centrality of a network node arises in many practical applications. In this paper, we study the optimization problem of maximizing the information centrality $I_v$ of a given node $v$ in a network with $n$ nodes and $m$ edges, by creating $k$ new edges incident to $v$. Since $I_v$ is the reciprocal of the sum of resistance distance $\mathcal{R}_v$ between $v$ and all nodes, we alternatively consider the problem of minimizing $\mathcal{R}_v$ by adding $k$ new edges linked to $v$. We show that the objective function is monotone and supermodular. We provide a simple greedy  algorithm with an approximation factor $\kh{1-\frac{1}{e}}$ and $O(n^3)$ running time. To speed up the computation, we also present an algorithm to compute $\left(1-\frac{1}{e}-\eps\right)$-approximate resistance distance $\mathcal{R}_v$ after iteratively adding $k$ edges, the running time of which is $\Otil (mk\eps^{-2})$ for any $\eps>0$, where the $\Otil (\cdot)$ notation suppresses the ${\rm poly} (\log n)$ factors. We experimentally demonstrate the effectiveness and efficiency of our proposed algorithms.
\end{abstract}

\section{Introduction}
Centrality metrics refer to indicators identifying the varying importance of nodes in complex networks~\cite{LuChReZhZhZh16}, which have become a powerful tool in network analysis and found wide applications in network science~\cite{Ne10}. Over the past years, a great number of centrality indices and corresponding algorithms have been proposed to analyze and understand the roles of nodes in networks~\cite{WhSm03,BoVi14}. Among various centrality indices, betweennees centrality and closeness centrality are probably the two most frequently used ones, especially in social network analysis. However, both indicators only consider the shortest paths, excluding the contributions from other longer paths. In order to overcome the drawback of these two measures, current flow closeness centrality~\cite{BrFl05,Ne05} was introduced and  proved to be exactly the information centrality~\cite{StZe89}, which counts all possible paths between nodes and has a better discriminating power than betweennees centrality~\cite{Ne05} and closeness centrality~\cite{BeWeLuMe16}.

It is recognized that centrality measures have proved of great significance in complex networks. Having high centrality can have positive consequences on the node itself. In this paper, we consider the problem of adding a given number of edges incident to a designated node $v$ so as to maximize the centrality of $v$. Our main motivation or justification for studying this problem is that it has several application scenarios, including airport networks~\cite{IsErTeBe12}, recommendation systems~\cite{PaPiTs16}, among others. For example, in airport networks, a node (airport) has the incentive to improve as much as possible its centrality (transportation capacity) by adding edges (directing flights) connecting itself and other nodes (airports)~\cite{IsErTeBe12}. Another example is the link recommendation problem of recommending to a user $v$ a given number of links from a set of candidate inexistent links incident to  $v$ in order to minimize the shortest distance from $v$ to other nodes~\cite{PaPiTs16}.

The problem of maximizing the centrality of a specific target node through adding edges incident to it has been widely studied. For examples, some authors have studied the problem of creating $k$ edges linked to a node $v$ so that the centrality value for $v$ with respect to concerned centrality measures is maximized, e.g., betweenness centrality~\cite{CrDaSeVe15,DaSeVe16,CrDaSeVe16,HoMoSo18} and closeness centrality~\cite{CrDaSeVe15,HoMoSo18}. Similar optimization problems for a predefined node $v$ were also addressed for other node centrality metrics, including average shortest distance between $v$ and remaining nodes~\cite{MeTa09,PaPiTs16}, largest distance from $v$ to other nodes~\cite{DeZa10}, PageRank~\cite{AvLi06,Os10}, and the number of different paths containing $v$~\cite{IsErTeBe12}. However, previous works do not consider improving information centrality of a node by adding new edges linked to it, despite the fact that it can better distinguish different nodes, compared with betweennees ~\cite{Ne05} and closeness centrality~\cite{BeWeLuMe16}.

In this paper, we study the following problem: Given a graph with $n$ nodes and $m$ edges, how to create $k$ new edges incident to a designated node $v$, so that the information centrality $I_v$ of $v$ is maximized. Since $I_v$ equals the reciprocal of the sum of resistance distance $\mathcal{R}_v$ between $v$ and all nodes, we reduce the problem to minimizing $\mathcal{R}_v$ by introducing $k$ edges connecting  $v$. We demonstrate that the optimization function is monotone and supermodular. To minimize resistance distance $\mathcal{R}_v$, we present two greedy approximation algorithms by iteratively introducing $k$ edges one by one. The former is a $\left(1-\frac{1}{e}\right)$-approximation algorithm with $O(n^3)$ time complexity, while the latter is a $\left(1-\frac{1}{e}-\eps\right)$-approximation algorithm with $\Otil (mk\eps^{-2})$ time complexity, where the $\Otil (\cdot)$ notation hides ${\rm poly} (\log n)$ factors. We test the performance of our algorithms on several model and real networks, which substantially increase information centrality score of a given node and outperform several other adding edge strategies.

\section{Preliminary}
Consider a connected undirected weighted network $G = (V,E,w)$ where $V$ is the set of nodes, $E \subseteq V \times V$ is the set of edges, and $w : E \to \rea_{+}$ is the edge weight function. We use $w_{max}$ to denote the maximum edge weight.
Let $n = |V|$ denote the number of nodes and $m = |E|$ denote the number of edges. For a pair of adjacent nodes $u$ and $v$, we write $u\sim v$ to denote $(u,v) \in E$.
The Laplacian matrix of $G$ is the symmetric matrix $\LL = \DD - \AA$, where $\AA$ is the weighted adjacency matrix of the graph and $\DD$ is the degree diagonal matrix.

Let $\ee_i$ denote the $i$th standard basis vector, and $\bb_{u,v} = \ee_u - \ee_v$.
We fix an arbitrary orientation for all edges in $G$.
For each edge $e\in E$, we define $\bb_e = \bb_{u,v}$, where $u$ and $v$
are head and tail of $e$, respectively. It is easy to verify that $\LL = \sum\nolimits_{e\in E} w(e)\bb_e \bb_e^\top$, where $w(e)\bb_e \bb_e^\top$ is the Laplacian of $e$. $\LL$ is  singular and positive semidefinite. Its pseudoinverse $\LL^\dag$ is $\kh{\LL +\frac{1}{n}\JJ}^{-1} - \frac{1}{n}\JJ$, where $\JJ$ is the matrix with all entries being ones.


For network $G = (V,E,w)$, the resistance distance~\cite{KlRa93} between two nodes $u,v$ is $\mathcal{R}_{uv} = \bb_{u,v}^\top \LL^\dag \bb_{u,v}$. The  resistance distance $\mathcal{R}_v$ of a node $v$ is the sum of resistance distances between  $v$ and all nodes in $V$, that is, $\mathcal{R}_v = \sum_{u \in V} \mathcal{R}_{uv}$, which can be expressed in terms of the entries of  $\LL^\dag$ as~\cite{BoFr13}
\begin{align}\label{eqR1}
\mathcal{R}_v = n \LL^\dag_{vv} + \trace{\LL^\dag}.
\end{align}
Let $\LL_v$ denote the submatrix of Laplacian $\LL$, which is obtained from $\LL$ by deleting the row and column corresponding to node $v$. For a connected graph $G$, $\LL_v$ is invertible for any node $v$,  and  the resistance distance $\mathcal{R}_{uv}$ between $v$ and another node $u$ is equal to $\kh{\LL_v^{-1}}_{uu}$~\cite{INK+13}. Thus, we have
\begin{align}\label{eqR2}
\mathcal{R}_v = \trace{\LL_v^{-1}}.
\end{align}


The  resistance distance $\mathcal{R}_v$ can be used as a  measure of the efficiency for node $v$ in transmitting information to other nodes, and is closely related to information centrality introduced by Stephenson and Zelen to measure the importance of nodes in social networks~\cite{StZe89}. The information $I_{uv}$ transmitted between
$u$ and $v$ is defined as
\begin{align*}
I_{uv} = \frac{1}{\BB^{-1}(u,u) + \BB^{-1}(v,v) - 2\BB^{-1}(u,v)},
\end{align*}
where $\BB = \LL + \JJ$. The information centrality $I_v$  of node $v$ is the harmonic mean of  $I_{uv}$ over all nodes $u$~\cite{StZe89}.
\begin{definition}
For a connected graph $G = (V,E,w)$, the information centrality $I_v$ of a node $v\in V$ is defined as
\begin{align*}
I_v = \frac{n}{ \sum\limits_{u\in V} 1 / I_{uv}}.
\end{align*}
\end{definition}
It was shown~\cite{BrFl05} that
\begin{align}\label{ItoR}
I_v = \frac{n}{\mathcal{R}_v}.
\end{align}

We continue to introduce some useful notations and tools for the convenience of description for our algorithms, including $\eps$-approximation and supermodular function.

Let $a,b\geq 0$ be two nonnegative scalars.
We say $a$ is an $\eps$-approximation~\cite{PS14} of $b$ if
$\exp(-\eps)\, a \leq b \leq \exp(\eps)\, a$. Hereafter, we use $a \approx_{\eps} b$ to represent that $a$ is an $\eps$-approximation of $b$.

Let $X$ be a finite set, and $2^X$ be the set of all subsets of $X$.  
Let $f: 2^X \to \mathbb{R}$ be a set function on $X$. For any subsets $S \subset T \subset X$ and any element $a \in X \setminus T$, we say function $f(\cdot)$ is supermodular if it satisfies $f(S) - f(S \cup \{a\}) \geq f(T) - f(T \cup \{a\})$. A function $f(\cdot)$ is submodular if $-f(\cdot)$ is supermodular.
A set function $f: 2^X \to \mathbb{R}$ is called monotone decreasing if for any subsets $S \subset T \subset X$,
 $f(S) > f(T)$ holds.

\section{Problem Formulation}

For a connected undirected weighted network $G(V,E,w)$, given a set $S$ of weighted edges not in $E$, we use $G(S)$ to denote the network augmented by adding the edges in $S$ to $G$, i.e. $G(S) = (V,E\cup S, w^\prime)$, where $w^\prime: E\cup S \to \rea_{+}$ is the new weight function. Let $\LL(S)$ denote the Laplacian matrix for $G(S)$. Note that the information centrality of a node depends on the graph topology.
If we augment a graph by adding a set of edges $S$,  the information centrality of a node will change. Moreover, adding edges incident to some node $v$ can only increase its information centrality~\cite{DoSn84}. \par

Assume that there is a set of nonexistent edges incident to a particular node $v$, each with a given weight. We denote this candidate edge set as $E_v$.
Consider choosing a subset $S$ of $k$ edges from the candidate set $E_v$ to augment the network so that the information centrality of  node $v$ is maximized. Let $I_v(S)$ denote the information centrality of the node $v$ in augmented network. We define the following set function optimization problem:
\begin{align}\label{pro1}
\underset{S\subset E_v,\, \sizeof{S}=k}{\operatorname{maximize}} \quad I_v(S).
\end{align}
Since the information centrality $I_v$ of a node $v$ is proportional to the reciprocal of $\mathcal{R}_v$, the optimization problem~(\ref{pro1}) is equivalent to  the following  problem:
\begin{align}\label{pro}
	\underset{S\subset E_v,\, \sizeof{S}=k}{\operatorname{minimize}} \quad \mathcal{R}_v(S),
\end{align}
where $\mathcal{R}_v(S)$ is the  resistance distance of $v$  in the augmented network $G(S)$. Without ambiguity, we take $\mathcal{R}(S)$ to replace $\mathcal{R}_v(S)$ for simplicity.\par

\section{Supermodularity of Objective Function}

Let $2^{E_v}$ denote all subsets of $E_v$. Then the resistance distance of  node $v$ in the augmented network can be represented as a set function $\mathcal{R} : 2^{E_v} \to \rea$.
To provide effective algorithms for the above-defined problems, we next prove that the resistance distance of $v$ is a supermodular function. \par

Rayleigh’s  monotonicity law~\cite{DoSn84} shows that the resistance distance between any pair of nodes can only decrease when edges are added. Then, we have the following theorem.
\begin{theorem}\label{thm:MI}
$\mathcal{R}(S)$ is a monotonically decreasing function of the set of edges $S$. That is, for any subsets $S \subset T \subset E_v$,
\begin{align*}
\mathcal{R}(T) < \mathcal{R}(S).
\end{align*}
\end{theorem}

We then prove the supermodularity of the objective function $\mathcal{R}(S)$.
\begin{theorem}\label{thm:SM}
$\mathcal{R}(S)$ is supermodular. For any set $S \subset T \subset E_v$ and any edge $e \in E_v \setminus T$,
\begin{align*}
\mathcal{R}(T) - \mathcal{R}(T \cup \{ e \} ) \leq \mathcal{R}(S) - \mathcal{R}(S \cup \{ e \} ).
\end{align*}
\end{theorem}

\begin{proof}
Suppose that edge $e$ connects two nodes $u$ and $v$, then $\LL(S\cup\{e\})_v = \LL(S)_v + w(e)\EE_{uu}$, where $\EE_{uu}$ is a square matrix with the $u$th diagonal entry being one, and all other  entries being zeros.
By~(\ref{eqR2}), it suffices to prove that
\small
\begin{align*}
	& \trace{\LL(T)_v^{-1}} - \trace{\kh{\LL(T)_v + w(e) \EE_{uu}}^{-1}} \\
	 \leq & \trace{\LL(S)_v^{-1}} - \trace{\kh{\LL(S)_v + w(e) \EE_{uu}}^{-1}}.
\end{align*}
\normalsize
Since $S$ is a subset of $T$, $\LL(T)_v = \LL(S)_v + \PP$, where $\PP$ is a nonnegative diagonal matrix. For simplicity, in the following proof, we use $\MM$ to denote  matrix $\LL(S)_v$. Then, we only need to prove
\small
\begin{align*}
	&\trace{\kh{\MM+\PP}^{-1}} - \trace{\MM^{-1}}  \\
	\leq &\trace{\kh{\MM+\PP + w(e) \EE_{uu}}^{-1}} - \trace{\kh{\MM + w(e) \EE_{uu}}^{-1}}.
\end{align*}
\normalsize

Define  function $f(t)$, $t \in [0,\infty)$, as
\small
\begin{align*}
f(t) = \trace{\kh{\MM+\PP + t \EE_{uu}}^{-1}} - \trace{\kh{\MM + t \EE_{uu}}^{-1}}.
\end{align*}
\normalsize
Then, the above inequality holds if $f(t)$ takes the minimum value at $t = 0$. We next show that  $f(t)$ is an increasing function by proving $\frac{df(t)}{dt} \geq 0$.
Using the matrix derivative formula  
\small
\begin{align*}
	\frac{d}{dt}\trace{\AA(t)^{-1}} = -\trace{\AA(t)^{-1} \frac{d}{dt}\AA(t) \AA(t)^{-1}},
\end{align*}
\normalsize we can differentiate function $f(t)$ as
\small
\begin{align*}
	\frac{df(t)}{dt}  = & - \trace{\kh{\MM+\PP + t \EE_{uu}}^{-1} \EE_{uu} \kh{\MM+\PP + t \EE_{uu}}^{-1}} \\
 				&+ \trace{\kh{\MM + t \EE_{uu}}^{-1} \EE_{uu} \kh{\MM + t \EE_{uu}}^{-1}}\\
 			      = & - \trace{ \EE_{uu} \kh{\MM+\PP + t \EE_{uu}}^{-2}} \\
 				&+ \trace{ \EE_{uu} \kh{\MM + t \EE_{uu}}^{-2}}\\
 			      =& - \kh{\kh{\MM+\PP + t \EE_{uu}}^{-2}}_{uu} + \kh{\kh{\MM + t \EE_{uu}}^{-2}}_{uu}.
\end{align*}
\normalsize

Let $\NN = \MM + t \EE_{uu}$,  and let $\QQ$ be a nonnegative diagonal matrix with exactly one positive diagonal entry $\QQ_{hh} > 0$ and all other entries being zeros. We now prove that $\NN^{-1}_{ij} \geq \kh{\NN+\QQ}^{-1}_{ij}$ for $1 \leq i,j \leq n-1$.  Using Sherman-Morrison formula~\cite{Me73}, we have
\begin{align*}
\NN^{-1} - \kh{\NN+\QQ}^{-1} = \frac{\QQ_{hh} \NN^{-1} \ee_h\ee_h^\top \NN^{-1}}{1+ \QQ_{hh} \ee_h^\top \NN^{-1} \ee_h}.
\end{align*}
Since $\NN$ is an M-matrix, every entry of $\NN^{-1}$ is positive~\cite{plemmons1977m}, it is the same with  every entry of $\NN^{-1} \ee_h\ee_h^\top \NN^{-1}$.  In addition,   the denominator $1+ \QQ_{hh} \ee_h^\top \NN^{-1} \ee_h$ is also positive, because $\NN$ is positive definite. Therefore, $\NN^{-1} - \kh{\NN+\QQ}^{-1}$ is a positive matrix,  the entries of which are all greater than zero.

By repeatedly applying the above process, we conclude  that $\NN^{-1} \geq \kh{\NN+\PP}^{-1}$  is a positive matrix. Thus,
\begin{align*}
\frac{df(t)}{dt} = - \kh{\kh{\NN+\PP}^{-2}}_{uu} + \kh{\NN^{-2}}_{uu} \geq 0,
\end{align*}
which completes the proof.
\end{proof}

\section{Simple Greedy Algorithm}

Theorems~\ref{thm:MI} and~\ref{thm:SM} indicate that  the objective function~(\ref{pro}) is a monotone and supermodular. Thus, a simple greedy algorithm is sufficient to approximate  problem~(\ref{pro}) with provable optimality bounds. In the greedy algorithm, the augmented edge set $S$ is initially empty. Then $k$ edges are iteratively added to the augmented edge set from the set $E_v$ of candidate edges. At each iteration, an edge $e_i$ in the candidate edge set is selected to maximize $\mathcal{R}(S) - \mathcal{R}(S \cup \{e_i\})$. The algorithm terminates when $|S| = k$. \par
According to~(\ref{eqR1}), the effective resistance $\mathcal{R}_v$ is equal to $n\LL_{vv}^\dag+\mathrm{Tr}(\LL^\dag)$. A naive algorithm requires $O(k |E_v|n^3)$ time complexity, which is prohibitively expense. Below we show that the computation cost can be reduced to $O(n^3)$ by using Sherman-Morrison formula~\cite{Me73}.
\begin{lemma} \label{lem:SMF}
For a connected weighted graph $G=(V,E,w)$ with weighted Laplacian matrix $\LL$, let $e$ be a nonexistent  edge with given weight $w(e)$ connecting  node   $v$. Then,
\small
\begin{align*}
\kh{\LL(\{e\})}^\dag = \kh{\LL + w(e)\bb_e\bb_e^\top}^\dag
= \LL^\dag - \frac{w(e) \LL^\dag \bb_e \bb_e^\top \LL^\dag}{1 + w(e)\bb_e^\top \LL^\dag \bb_e}.
\end{align*}
\normalsize
\end{lemma}

For a candidate edge not added to  $S$, let  $\mathcal{R}_v^\Delta(e) = \mathcal{R}(S) - \mathcal{R}(S \cup \{e\})$. Lemma~\ref{lem:SMF} and~(\ref{eqR1}) lead to the  following result.
\begin{lemma}\label{lem:SM}
Let $G=(V,E,w)$ be a connected weighted graph with weighted Laplacian matrix $\LL$. Let  $e \not \in E$ be a candidate edge with given weight $w(e)$  incident to node $v$. Then,
\small
\begin{align}\label{eq:SM}
	\mathcal{R}_v^\Delta(e) =
	\frac{w(e) \kh{n\kh{\LL^\dag \bb_e \bb_e^\top \LL^\dag}_{vv}+\trace{\LL^\dag \bb_e \bb_e^\top \LL^\dag}}}{1+ w(e) \bb_e^\top \LL^\dag \bb_e}.
\end{align}
\normalsize
\end{lemma}

Lemma~\ref{lem:SM} yields a  simple greedy algorithm $\ExactSM(G, v, E_v , k)$, as outlined in Algorithm~\ref{alg:MVC1}.  	The first step of this algorithm is to  compute the pseudoinverse of $\LL$, the time complexity of which is $O(n^3)$ time. Then this algorithm works in $k$ rounds, each  involving operations of computations and updates with time complexity $O(n^2)$. Thus, the total running time of Algorithm~\ref{alg:MVC1} is $O(n^3)$.
\begin{algorithm}
	\caption{$\ExactSM(G, v, E_v , k)$}
	\label{alg:MVC1}
	\Input{
		A connected graph $G$; a node $v \in V$; a candidate edge set $E_v$;
		an integer $k \leq |E_v|$
	}
	\Output{
		A subset of $S \subset E_v$ and $|S| = k$
	}
	Initialize solution $S = \emptyset$ \;
	Compute $\LL^\dag$ \;
	\For{$i = 1$ to $k$}{
		Compute $\mathcal{R}_v^\Delta(e)$
 for each $e \in E_v \setminus S$ \;
		Select $e_i$ s.t. $e_i \gets \mathrm{arg\,max}_{e \in E_v \setminus S} \mathcal{R}_v^\Delta(e)$ \;
		Update solution $S \gets S \cup \{ e_i \}$ \;
		Update the graph $G \gets G(V,E \cup \{ e_i \})$ \;
		Update \resizebox{0.6\hsize}{!}{$\LL^\dag \gets \LL^\dag- \frac{w(e_i) \LL^\dag \bb_{e_i} \bb_{e_i}^\top \LL^\dag}{1 + w(e_i)\bb_{e_i}^\top \LL^\dag \bb_{e_i}}$}
	}
    \Return S \;
\end{algorithm}

Moreover, due to the result in~\cite{nemhauser1978analysis}, Algorithm~\ref{alg:MVC1} is able to achieve a $\left(1 - \frac{1}{e}\right)$  approximation factor, as given in the following theorem.
\begin{theorem}
The set $S$ returned by Algorithm~\ref{alg:MVC1} satisfies
\small
\begin{align*}\label{eq:bound1}
	\mathcal{R}(\emptyset) - \mathcal{R}(S) \geq \kh{1 - \frac{1}{e}}(\mathcal{R}(\emptyset) - \mathcal{R}(S^*)),
\end{align*}
\normalsize
where $S^*$ is the optimal solution to~(\ref{pro}), i.e.,
\small
\begin{align*}
	S^* \, \, \defeq \underset{S\subset V,\, \sizeof{S}=k}{\mathrm{arg\,min}} \quad \mathcal{R}(S).
\end{align*}
\normalsize
\end{theorem}


\section{Fast Greedy Algorithm}

Although Algorithm~\ref{alg:MVC1} is faster than the naive algorithm, it is still  computationally infeasible for large networks,   since it involves the computation of the pseudoinverse for $\LL$. In this section, in order to avoid inverting the matrix $\LL$, we  give an efficient  approximation algorithm, which achieves a  $\left(1 - \frac{1}{e}-\eps\right)$  approximation factor of optimal solution to  problem~(\ref{pro}) in time $\Otil(km\eps^{-2})$.

\subsection{Approximating $\mathcal{R}_v^\Delta(e)$}

In order to solve problem~(\ref{pro}), one need to compute the key quantity  $\mathcal{R}_v^\Delta(e)$ in~\eqref{eq:SM}. Here, we provide an efficient algorithm to approximate  $\mathcal{R}_v^\Delta(e)$ properly. \par

We first consider the denominator in~\eqref{eq:SM}. Assume that the new added edge  $e$ connects  nodes $u$ and $v$.
Note that the term $r_e = \bb_e^\top \LL^\dag \bb_e$  in the denominator is in fact the  resistance distance $\mathcal{R}_{uv}$ between $u$ and $v$ in the network excluding $e$.  It can be computed by the following  approximation
algorithm~\cite{SS11}.

\begin{lemma}\label{lem:ERest}
Let $G = (V,E,w)$ be a weighted connected graph. There is an algorithm $\textsc{ApproxiER}(G, E_v,\eps)$ that
returns an estimate $\hat{r}_e$ of $r_e$ for all $e\in E_v$ in
	$\Otil(m \eps^{-2})$ time.
	With probability at least $1-1/n$, $\hat{r}_e \approx_{\eps} r_e$ holds
	for all $e\in E_v$.
\end{lemma}

For the numerator of~\eqref{eq:SM}, it includes two terms, $\kh{\LL^\dag \bb_e \bb_e^\top \LL^\dag}_{vv}$ and $\trace{\LL^\dag \bb_e \bb_e^\top \LL^\dag}$. The first term can be calculated by $\kh{\LL^\dag \bb_e \bb_e^\top \LL^\dag}_{vv} = \ee_v^\top \LL^\dag \bb_e \bb_e^\top \LL^\dag \ee_v$. The second term is the trace of an implicit matrix which can be approximated by Hutchinson's Monte-Carlo method~\cite{Hut89}. By generating $M$ independent random $\pm 1$ vectors $\xx_1, \xx_2, \cdots ,\xx_M \in \mathbb{R}^n$ (i.e., independent Bernoulli entries), $\frac{1}{M}\sum\nolimits_{i=1}^M \xx_i^\top \AA \xx_i$ can be used to estimate the trace of matrix $\AA$.
Since $\expec{}{\xx_i^\top \AA \xx_i} = \trace{\AA}$, by the law of large numbers,
$\frac{1}{M}\sum\nolimits_{i=1}^M \xx_i^\top \AA \xx_i$ should be close to $\trace{\AA}$ when $M$ is large.
The following lemma~\cite{AT11}  provides a good estimation of $\trace{\AA}$.
\begin{lemma}
\label{lem:mcL}
	Let $\AA$ be a positive semidefinite matrix with rank $\mathrm{rank}(\AA)$.
	Let $\xx_1,\ldots,\xx_M$ be independent random $\pm 1$ vectors.
	Let $\epsilon, \delta$ be scalars such that $0 < \epsilon \leq 1/2$
	and $0 < \delta < 1$.
	For any $M \geq 24\epsilon^{-2} \ln(2 \mathrm{rank}(\AA) / \delta)$,
	the following statement holds with probability at least $1 - \delta$:
\small
	\begin{align*}
		\frac{1}{M}\sum\limits_{i=1}^M \xx_i^\top \AA \xx_i \approx_\epsilon \trace{\AA}.
	\end{align*}
\normalsize
\end{lemma}


Thus, we have reduced  the estimation of the numerator of~\eqref{eq:SM} to the calculation of the quadratic form of $\LL^\dag \bb_e \bb_e^\top \LL^\dag$. If we directly compute the quadratic form, we must first evaluate $\LL^\dag$, the time complexity is high.   To avoid inverting $\LL$, we will utilize the nearly-linear time solver for Laplacian systems from~\cite{kyng2016approximate}, whose
performance can be characterized in the following lemma.

\begin{lemma}
	\label{lem:laplsolver}
	The algorithm $\yy = \LaplSolver(\LL,\zz,\eps)$ takes a Laplacian matrix $\LL$ of a graph $G$ with $n$ nodes and $m$ edges, a vector $\zz \in \rea^n$ and a scalar $\eps > 0$ as input, and returns a vector $\yy \in \rea^n$ such that with probability $1-1/{\rm poly}(n)$ the following statement holds:
\small
	\begin{align*}
		\len{\yy - \LL^\dag \zz}_{\LL} \leq \eps
		\len{\LL^\dag \zz}_{\LL},
	\end{align*}
\normalsize
	where $\len{\xx}_{\LL} = \sqrt{\xx^\top \LL \xx}$.
	The algorithm runs in expected time
	$\Otil(m)$.
\end{lemma}


Lemmas~\ref{lem:ERest},~\ref{lem:mcL} and~\ref{lem:laplsolver} result in the  following algorithm $\VReffComp(G, v, E_v, \epsilon)$ for computing  $\mathcal{R}_v^\Delta(e)$ for all $ e \in E_v$, as depicted in Algorithm~\ref{alg:vreffcomp}. The algorithm has a total running time  	$\Otil(m\epsilon^{-2})$, and returns a set of pairs $\{(e, \hat{\mathcal{R}}_v^\Delta(e)) | e \in E_v\}$, satisfying that $\mathcal{R}_v^\Delta(e) \approx_{\epsilon} \hat{\mathcal{R}}_v^\Delta(e)$ for all $ e \in E_v$.

\begingroup
\begin{algorithm}
\small
	\caption{$\VReffComp(G, v, E_v, \epsilon)$}
	\label{alg:vreffcomp}
	\Input{
		A graph $G$; a node $v \in V$; a candidate edge set $E_v$; a real number $0 \leq \epsilon \leq1/2$ \\
	}
	\Output{
		$\{(e,\hat{\mathcal{R}}_v^\Delta(e)) | e \in E_v\}$
	}
	Let $\zz_1,\ldots,\zz_M$ be independent random $\pm 1$ vectors, where $M = \ceil{432\eps^{-2}\ln(2n)}$. \;
	\For{$i = 1$ to $M$}{
		$\yy_i \gets \LaplSolver(\LL, \zz_i, \frac{1}{72}\eps n^{-8} w_{max}^{-4})$ \;
		\For{each $e \in E_v$}{
			Compute $t_i(e) \defeq \yy_i^\top \bb_e \bb_e^\top \yy_i$ \;
		}
	}
	$\xx \gets \LaplSolver(\LL, \ee_v, \frac{1}{72}\eps n^{-9} w_{max}^{-4})$ \;
	\For{each $e \in E_v$}{
	Compute $\alpha(e) \defeq \xx^\top \bb_e \bb_e^\top \xx$ \;
	}
	$\hat{r}_e \gets \textsc{ApproxiER}(G, \eps/3)$ \;
	Compute \resizebox{0.6\hsize}{!}{$\hat{\mathcal{R}}_v^\Delta(e) = w(e) \frac{n\alpha(e) + \frac{1}{M} \sum\limits_{i}^{M} t_i(e)}{1 + w(e)\hat{r}_e}$} for each $e$ \;
    \Return $\{(e,\hat{\mathcal{R}}_v^\Delta(e)) | e \in E_v\}$
\normalsize
\end{algorithm}
\endgroup

\subsection{Fast  Algorithm for Objective Function }

By using Algorithm~\ref{alg:vreffcomp} to approximate $\mathcal{R}_v^\Delta(e)$, we give a fast greedy algorithm $\ApproxiSM(G, v, E_v, k, \epsilon)$ for solving problem~(\ref{pro}), as outlined in Algorithm~\ref{alg:MVC2}.

\begingroup
\begin{algorithm}
\small
	\caption{$\ApproxiSM(G, v, E_v, k, \epsilon)$}
	\label{alg:MVC2}
	\Input{
		A graph $G$; a node $v \in V$; a candidate edge set $E_v$; an integer $k \leq |E_v|$; a real number $0 \leq \epsilon \leq1/2$
	}
	\Output{
		$S$: a subset of $E_v$ and $|S| = k$
	}
	Initialize solution $S = \emptyset$ \;
	\For{$i = 1$ to $k$}{
		$\{e,\hat{\mathcal{R}}_v^\Delta(e) | e \in E_v \setminus S \} \gets \VReffComp(G, v, E_v \setminus S, 3\epsilon)$. \;
		Select $e_i$ s.t. $e_i \gets \mathrm{arg\,max}_{e \in E_v \setminus S} \hat{\mathcal{R}}_v^\Delta(e)$ \;
		Update solution $S \gets S \cup \{ e_i \}$ \;
		Update the graph $G \gets G(V,E \cup \{ e_i \})$
	}
	\Return $S$
\normalsize
\end{algorithm}
\endgroup

Algorithm~\ref{alg:MVC2} works in $k$ rounds (Lines 2-6). In every round, the call of $\VReffComp$ and updates take time $\Otil(m\eps^{-2})$. Then, the total running time of  Algorithm~\ref{alg:MVC2} is $\Otil(km\eps^{-2})$.  The following theorem shows that the output $\hat{S}$ of  Algorithm~\ref{alg:MVC2} gives a $\kh{1 - \frac{1}{e} - \eps}$ approximate solution to problem~(\ref{pro}).
\begin{theorem}
For any  $0<\epsilon \leq1/2$, the set $\hat{S}$ returned by the greedy algorithm above satisfies
\small
\begin{align*}
\mathcal{R}(\emptyset) - \mathcal{R}(\hat{S}) & \geq \kh{1 - \frac{1}{e} - \eps} (\mathcal{R}(\emptyset) - \mathcal{R}(S^*)),
\end{align*}
\normalsize
where $S^*$ is the optimal solution to problem~(\ref{pro}), i.e.,
\small
\begin{align*}
	S^* \, \, \defeq \underset{S\subset V,\, \sizeof{S}=k}{\mathrm{arg\,min}} \quad \mathcal{R}(S).
\end{align*}
\normalsize
\end{theorem}

We omit the proof, since it is similar to that in~\cite{BaVo14}.

\section{Experiments}

In this section, we experimentally evaluate the effectiveness and efficiency of our two greedy algorithms on some model and real  networks. All algorithms in our experiments are implemented in Julia. In our algorithms, we use the  $\LaplSolver$~\cite{kyng2016approximate}, the  implementation (in Julia) of  which is available on website\footnote{https://github.com/danspielman/Laplacians.jl}. All experiments were conducted on a machine with 4.2 GHz Intel i7-7700 CPU and 32G RAM.

We execute our experiments on two popular model networks, Barab\'asi-Albert (BA) network and Watts–Strogatz (WS) network, and a large connection of realistic networks from KONECT~\cite{kunegis2013konect} and SNAP\footnote{https://snap.stanford.edu}. Table~\ref{SetNo} provides the information of these networks, where real-world networks are shown in increasing size of the  number of nodes in original networks.

\begin{table}[htbp]
\setlength{\abovecaptionskip}{0.pt}
\caption{Statistics of datasets. For a network with $n$ nodes and $m$ edges, we denote the number of nodes and edges in its largest connected component by $n'$ and $m'$, respectively. }\label{SetNo}
\begin{center}
\normalsize
\resizebox{0.8\columnwidth}{!}{
\begin{tabular}{ccccc}
\Xhline{2.5\arrayrulewidth}
\raisebox{-0.5ex}{Network} & \raisebox{-0.5ex}{$n$} & \raisebox{-0.5ex}{$m$} & \raisebox{-0.5ex}{$n'$} & \raisebox{-0.5ex}{$m'$} \\[0.5ex]
\hline
\raisebox{-0.5ex}{BA network} & \raisebox{-0.5ex}{50} & \raisebox{-0.5ex}{94} & \raisebox{-0.5ex}{50} & \raisebox{-0.5ex}{94} \\[0.5ex]
\raisebox{-0.5ex}{WS network} & \raisebox{-0.5ex}{50} & \raisebox{-0.5ex}{100} & \raisebox{-0.5ex}{50} & \raisebox{-0.5ex}{100} \\[0.5ex]
\raisebox{-0.5ex}{Zachary karate club} & \raisebox{-0.5ex}{34} & \raisebox{-0.5ex}{78} & \raisebox{-0.5ex}{34} & \raisebox{-0.5ex}{78} \\[0.5ex]
\raisebox{-0.5ex}{Windsufers} & \raisebox{-0.5ex}{43} & \raisebox{-0.5ex}{336} & \raisebox{-0.5ex}{43} & \raisebox{-0.5ex}{336} \\[0.5ex]
\raisebox{-0.5ex}{Jazz musicians} & \raisebox{-0.5ex}{198} & \raisebox{-0.5ex}{2742} & \raisebox{-0.5ex}{195} & \raisebox{-0.5ex}{1814} \\[0.5ex]
\raisebox{-0.5ex}{Virgili} & \raisebox{-0.5ex}{1,133} & \raisebox{-0.5ex}{5,451} & \raisebox{-0.5ex}{1,133} & \raisebox{-0.5ex}{5,451} \\[0.5ex]
\raisebox{-0.5ex}{Euroroad} & \raisebox{-0.5ex}{1,174} & \raisebox{-0.5ex}{1,417} & \raisebox{-0.5ex}{1,039} & \raisebox{-0.5ex}{1,305} \\[0.5ex]
\raisebox{-0.5ex}{Hamster full} & \raisebox{-0.5ex}{2,426} & \raisebox{-0.5ex}{16,631} & \raisebox{-0.5ex}{2,000} & \raisebox{-0.5ex}{16,098} \\[0.5ex]
\raisebox{-0.5ex}{Facebook} & \raisebox{-0.5ex}{2,888} & \raisebox{-0.5ex}{2,981} & \raisebox{-0.5ex}{2,888} & \raisebox{-0.5ex}{2,981} \\[0.5ex]
\raisebox{-0.5ex}{Powergrid} & \raisebox{-0.5ex}{4,941} & \raisebox{-0.5ex}{6,594} & \raisebox{-0.5ex}{4,941} & \raisebox{-0.5ex}{6,594} \\[0.5ex]
\raisebox{-0.5ex}{ca-GrQc} & \raisebox{-0.5ex}{5,242} & \raisebox{-0.5ex}{14,496} & \raisebox{-0.5ex}{4,158} & \raisebox{-0.5ex}{13,422} \\[0.5ex]
\raisebox{-0.5ex}{ca-HepPh} & \raisebox{-0.5ex}{12,008} & \raisebox{-0.5ex}{118,521} & \raisebox{-0.5ex}{11,204} & \raisebox{-0.5ex}{117,619} \\[0.5ex]
\raisebox{-0.5ex}{com-DBLP} & \raisebox{-0.5ex}{317,080} & \raisebox{-0.5ex}{1,049,866} & \raisebox{-0.5ex}{317,080} & \raisebox{-0.5ex}{1,049,866} \\[0.5ex]
\raisebox{-0.5ex}{roadNet-TX} & \raisebox{-0.5ex}{1,379,917} & \raisebox{-0.5ex}{1,921,660} & \raisebox{-0.5ex}{1,351,137} & \raisebox{-0.5ex}{1,879,201} \\[0.5ex]
\Xhline{2.5\arrayrulewidth}
\end{tabular}
}
\vspace{-10pt}
\end{center}
\end{table}
\subsection{Effectiveness of Greedy Algorithms}

To show the effectiveness of our algorithms,  we compare the results of our algorithms with the optimum solutions on two small model networks, BA network and   WS network, and two  small  real-world networks, Zachary karate club network and Windsufers contact network. Since these networks are small, we are able to compute the optimal edge set.

For each network, we randomly choose 20 target nodes. For each target node $v$,  the candidate edge set is composed of all nonexistent edges incident to it with unit weight $w = 1$.  And for each  designated  $k = 1,2,\cdots,6$, we add  $k$ edges linked to $v$ and other $k$ non-neighboring nodes of $v$. We then compute the average  information centrality of the 20 target nodes for each $k$. Also, we compute the solutions for the random scheme, by adding $k$ edges from randomly selected $k$  non-neighboring nodes.  The results are reported in Figure~\ref{ComOpt}. We observe that there is little difference between the solutions of our greedy algorithms and the optimal solutions, since their approximation ratio  is always greater than 0.98, which is far better than the theoretical guarantees. Moreover, our greedy schemes outperform the random scheme in these four networks.

%

\begin{figure}
\centering
\includegraphics[width=0.45\textwidth]{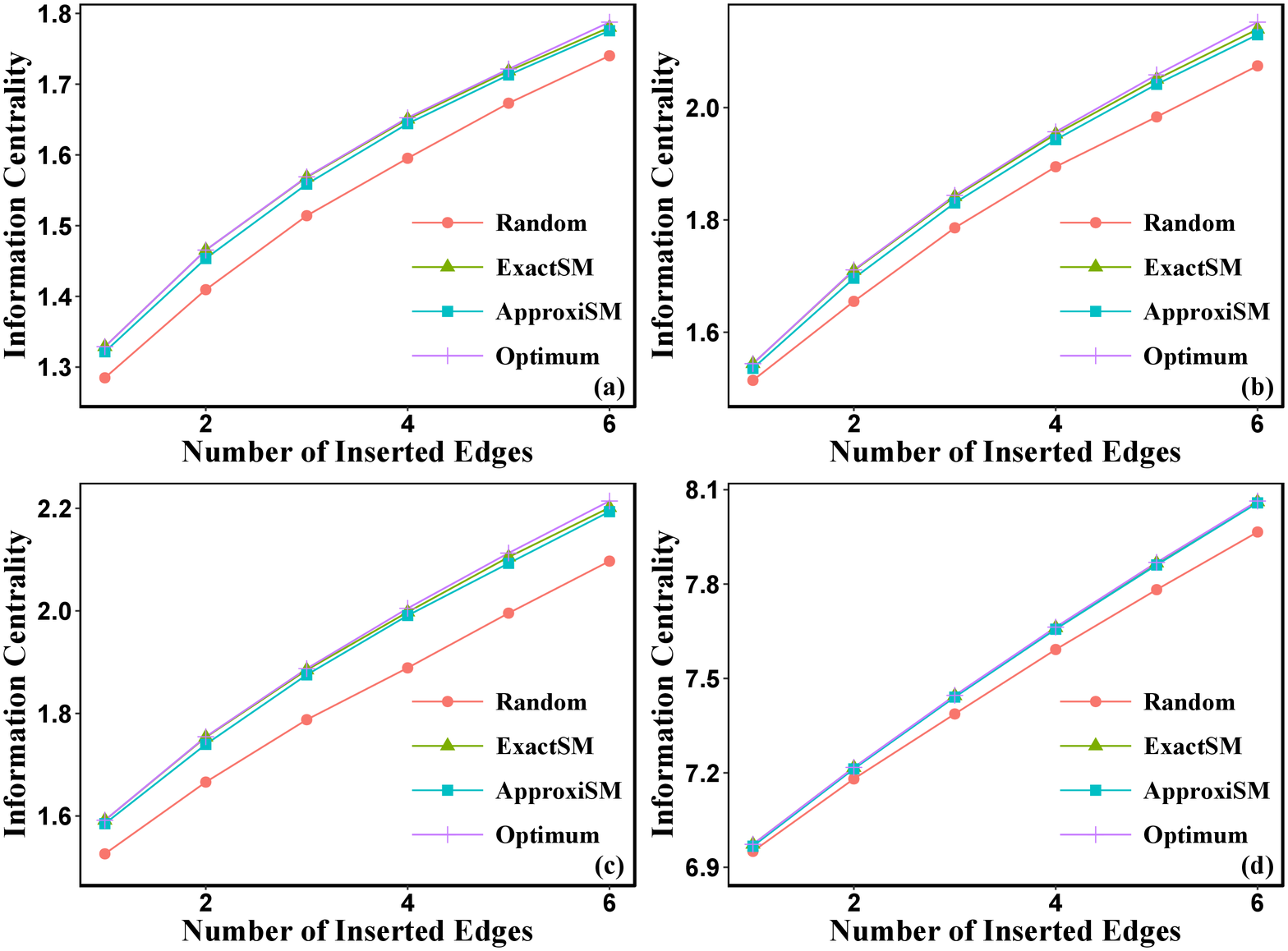}
\caption{Average information centrality of target nodes as a function of the number $k$ of inserted edges for $\ExactSM$, $\ApproxiSM$, random and the optimum solution on four networks: BA  (a), WS  (b), Karate club (c),  and Windsufers (d).\label{ComOpt}}
\end{figure}

To further demonstrate  effectiveness of our algorithms, we  compare the results of our methods with the random scheme and other two baseline schemes, Top-degree and Top-cent, on four other real-world networks. In Top-degree scheme, the added edges are simply the  $k$ edges  connecting target node $v$  and its nonadjacent nodes with the highest degree in the original network; while in Top-cent scheme, the added edges are simply those  $k$ edges  connecting target node $v$ and its nonadjacent  nodes with the largest information centrality  in the original network.

Since the results may vary depending on the initial information centrality of the target node $v$, for each of the four real networks, we select 10 different target nodes at random.  For each target node, we first compute its original information centrality and  increase
it  by adding up to $k = 20$ new edges, using our two greedy algorithms and the three baselines. Then, we compute and record the information centrality of the target node after insertion of every edge. Finally, we compute the average  information centrality of all the 10 target nodes for each   $k =1,2,\ldots, 20$, which is plotted in  Figure~\ref{ComBase}.  We  observe that for all the four real-world networks  our greedy algorithms outperform the three baselines.

\begin{figure}
\centering
\includegraphics[width=.45\textwidth]{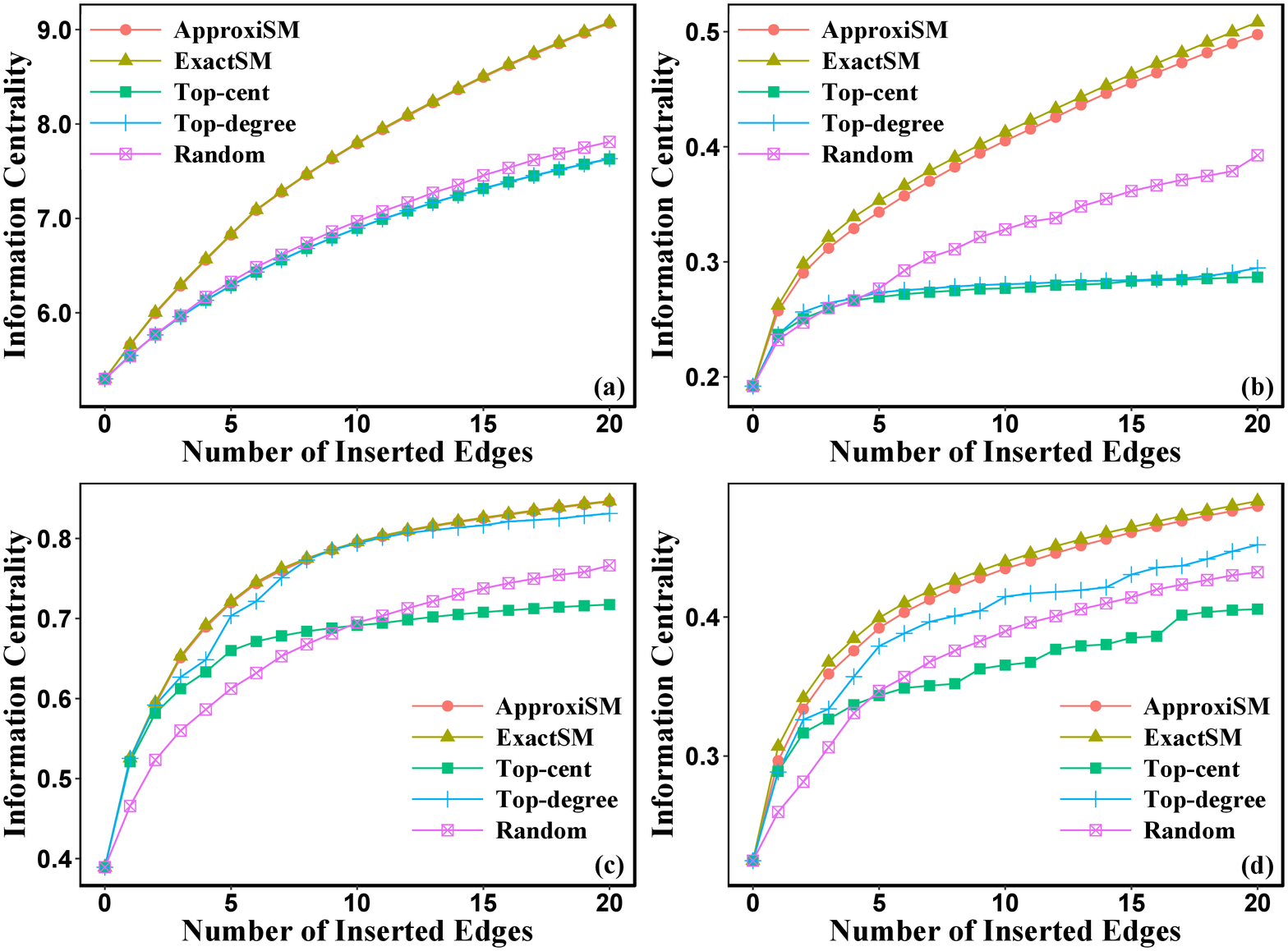}
\caption{Average information centrality of target nodes as a function of the number $k$ of inserted edges for the five heuristics on Jazz musicians (a), Euroroad (b), Facebook (c), Powergrid (d).\label{ComBase} }
\end{figure}

\subsection{Efficiency Comparison of Greedy Algorithms}

Although both of our  greedy algorithms are effective, we will show that their efficiency greatly differs. To this end, we compare the efficiency of the greedy algorithms on several  real-world networks. For each network, we choose stochastically 20 target nodes,  for each of which, we create $k = 10$ new edges incident to it to maximize its information centrality according to Algorithms~\ref{alg:MVC1} and~\ref{alg:MVC2}.  We compute the average  information centrality of  10 target nodes for each network and record the average running times. In Table~\ref{time} we provide the results of average  information centrality and  average running time of our greedy algorithms.  We observe that $\ApproxiSM$ algorithm are faster than $\ExactSM$ algorithm, especially for large  networks,  while their final information centrality score are close. More interestingly, $\ApproxiSM$ applies to massive networks.  For example, for com-DBLP  and roadNet-TX networks, $\ApproxiSM$ computes  their information centrality in half an hour, while $\ApproxiSM$ fails due to its high time complexity.

\begin{table}[tp]
\setlength{\abovecaptionskip}{5.pt}
\setlength{\belowcaptionskip}{-0.cm}
  \centering
  \fontsize{6.5}{8}\selectfont
  \begin{threeparttable}
  \caption{The average running times and results of $\ApproxiSM$ (ASM) and $\ExactSM$ (ESM) algorithms on several real-world networks, as well as  the ratios for times and results of $\ApproxiSM$ to those  of $\ExactSM$.\label{time}}
  \label{tab:performance_comparison}
    \begin{tabular}{ccccccc}
    \toprule
    \multirow{2}{*}{Network}&
    \multicolumn{3}{c}{Time (seconds)}&\multicolumn{3}{c}{Information centrality}\cr
    \cmidrule(lr){2-4} \cmidrule(lr){5-7}
    &ASM&ESM&Ratio&ASM&ESM&Ratio\cr
    \midrule
    Virgili              &1.3996  & 0.9172  & 1.5259 & 2.5005 & 2.5037 & 0.9987\cr
    Euroroad             &0.6563  & 0.7593  & 0.8643 & 0.4003 & 0.4069 & 0.9838\cr
    Hamster  full        &3.0785  & 4.8528  & 0.6344 & 2.9904 & 2.9944 & 0.9987\cr
    Facebook             &1.7151  & 12.9203 & 0.1327 & 0.7937 & 0.7947 & 0.9987\cr
    Powergrid            &5.8727  & 58.3359 & 0.1006 & 0.4327 & 0.4369 & 0.9904\cr
    ca-GrQc              &5.3023  & 34.0228 & 0.1558 & 1.2118 & 1.2136 & 0.9985\cr
    ca-HepPh             &28.7462 & 620.4557& 0.0463 & 2.2569 & 2.2592 & 0.9990\cr
    com-DBLP             &697.1835& -       & -      & 1.1327 & -      & -     \cr
    roadNet-TX           &1569.5059&-       &-       & 0.0556 & -      & -     \cr
    \bottomrule
    \end{tabular}
    \end{threeparttable}

    \vspace{-5pt}
\end{table}

\section{Conclusions}

In this paper, we considered the problem of maximizing the information centrality of a designated node $v$ by adding $k$ new edges incident to it. This problem is equivalent to minimizing the resistance distance $\mathcal{R}_v$ of node $v$. We proposed two approximation algorithms for computing $\mathcal{R}_v$ when $k$ edges are repeatedly inserted in a greedy way. The first one gives a $\left(1-\frac{1}{e}\right)$ approximation of the optimum in time $O(n^3)$.
While the second one returns a $\left(1-\frac{1}{e}-\eps\right)$ approximation in time $\Otil (mk\eps^{-2})$. Since the considered problem has never addressed before, we have no other algorithms to compare with, but compare our algorithms with potential alternative algorithms. Extensive experimental results on model and realistic networks show that our algorithms can often compute an approximate optimal solution. Particularly, our second algorithm can achieve a good approximate solution very quickly, making it applicable to massive networks.


\end{document}